\documentclass[12pt,draftcls,onecolumn]{IEEEtran}
\usepackage{nonfloat}
\usepackage{graphicx}
\usepackage{amssymb}
\usepackage{booktabs}
\usepackage{epstopdf}
\usepackage{amsmath}
\usepackage{amsthm}
\usepackage{multirow}
\usepackage{caption}
\usepackage{subfig}
\usepackage{algorithm2e}
\newtheorem{mydef}{Definition}

\newtheorem{thm}{Theorem}

\begin{document}

\title{Privacy Preserving Recommendation System Based on Groups}
       \author{
\IEEEauthorblockN{Shang Shang\IEEEauthorrefmark{1},  Yuk Hui\IEEEauthorrefmark{2}, Pan
Hui\IEEEauthorrefmark{3},
Paul Cuff\IEEEauthorrefmark{1},  Sanjeev Kulkarni\IEEEauthorrefmark{1} }\\ \IEEEauthorblockA{\IEEEauthorrefmark{1}Department of Electrical Engineering, 
Princeton University \\ Princeton NJ, 08540, U.S.A.
} \\ \IEEEauthorblockA{\IEEEauthorrefmark{2}
 Centre for Digital Cultures, Leuphana University,  L\"{u}neburg, Germany}\\
\IEEEauthorblockA{\IEEEauthorrefmark{3}
Department of Computer Science, The Hong Kong University of Science and Technology, Hong Kong, China}\\
\IEEEauthorrefmark{1}{\{sshang, cuff, kulkarni\}@princeton.edu, \IEEEauthorrefmark{2}yuk.hui@inkubator.leuphana.de, \IEEEauthorrefmark{3}panhui@cse.ust.hk}
}


\maketitle
\begin{abstract}
Recommendation systems have received considerable attention in the recent decades. Yet with the development of information technology and social media, the risk in revealing private data to service providers has been a growing concern to more and more users.  Trade-offs between quality and privacy in recommendation systems naturally arise. In this paper, we present a privacy preserving recommendation framework based on groups. The main idea is to use groups as a natural middleware to preserve users' privacy. A distributed preference exchange algorithm is proposed to ensure the anonymity of data, wherein the effective size of the anonymity set asymptotically approaches the group size with time. We construct a hybrid collaborative filtering model based on Markov random walks to provide recommendations and predictions to group members. Experimental results on the MovieLens and Epinions datasets show that our proposed methods outperform the baseline methods, L+ and ItemRank, two state-of-the-art personalized recommendation algorithms, for both recommendation precision and hit rate despite the absence of personal preference information.
 
\end{abstract}


\section{Introduction}
With the recent development of social media, personalization and privacy preservation are often in tension with each other. Private companies such as Google and Facebook are accumulating and recording enormous personal data for the sake of personalization. Personalization provides users with conveniences. At the same time, it can have a direct impact on marketing, sales, and profit. Most recommendation systems focus on improving the performance of collaborative filtering (CF) techniques. Privacy, which is a serious concern for many users, is the price users have to pay for the convenience of recommendation systems in a world with booming information. Users normally have no choice but to trust the service provider to keep their sensitive personal profile safe. However, it is not always ``safe.'' For example, a shopping website one has visited once might keep appearing on the advertising block for days when browsing some other web pages. 

The starting point of our paper is to find a way out of the opposition between anonymity and personalization: how can we maintain a certain level of anonymity without sacrificing useful and accurate recommendations? We propose to do recommendations at a group level, instead of at the individual level. Group based social networks (for example, Diaspora, Crabgrass, Lorea, etc.) were originally conceived as alternatives for social networks such as Facebook, twitter, etc, and are gaining more and more users \cite{Yuk}. Also, group-based social networks have been thriving on the other side of the globe, notably Douban (as shown in Fig.\ref{douban}), a Chinese group-based social network focus on building interest groups around books, films, music, etc., has already more than 50 million users. The Douban example demonstrates that these group-based models are not simply of marginal interest. As privacy issues generating increasing concern, alternative designs such as group-based social networks may continue to emerge. This departure from individual based social networks to group based social networks inspired this study. We find that it is possible to give accurate recommendations based on groups while maintaining some privacy from the service provider.			

\begin{figure*}[!t]
\centering
\subfloat[An example of a group-based social network: douban.com. On the left of the webpage is shown information of a DIY group, and on the right is shown a list of new-coming group members and associated groups.]{\includegraphics[width=0.6\textwidth]{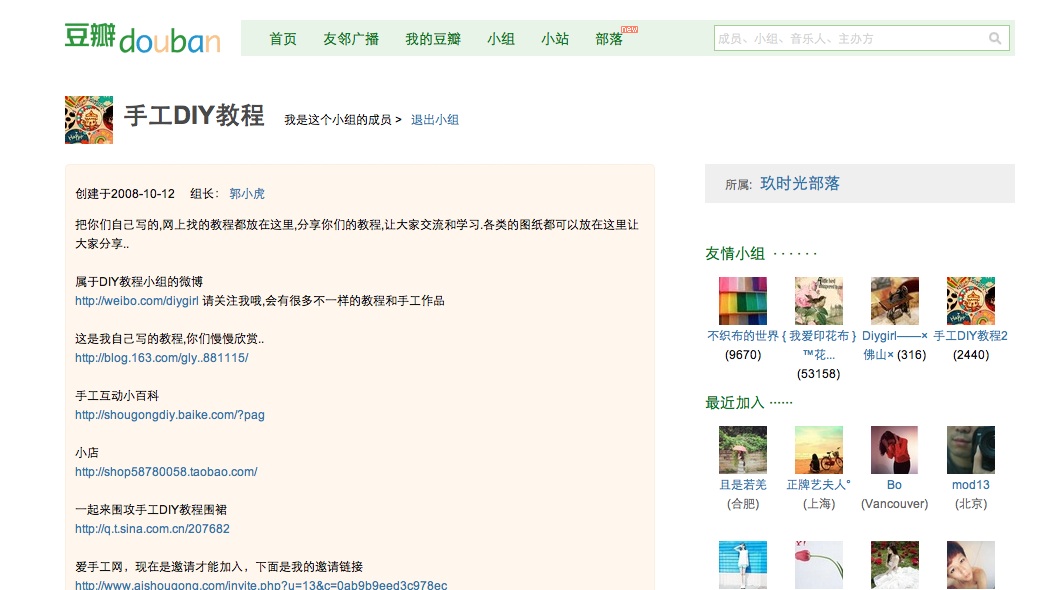}}\hfill
\subfloat[Structure of group-based social networks. Two groups are linked if they are associate groups.]{\includegraphics[width=0.35\textwidth]{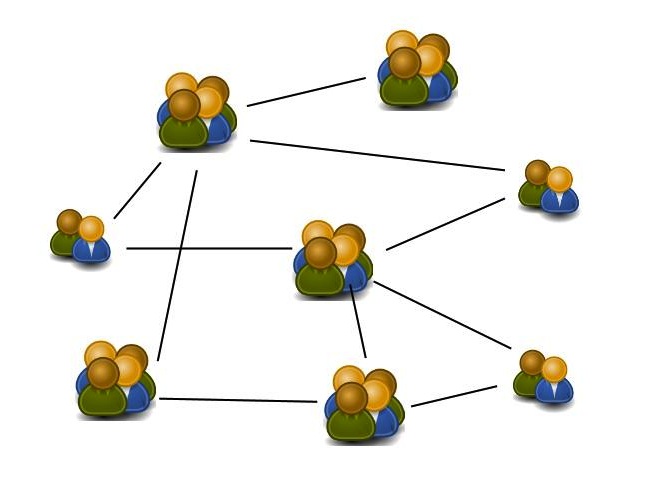}}
\caption{Group based social networks. }
\label{douban}
\end{figure*}

\subsection{Related Work}

Current approaches to protect privacy in recommendation systems mostly address two different privacy concerns: protecting users' privacy from curious peers or malicious users \cite{Ash, Frank}, and against unreliable service providers \cite{Aimeur, Canny, P3}. In order to make the outcome of recommendation insensitive to single input so as to protect users private preference data from other users, privacy preserving algorithms from the differential privacy literature are modified to provide privacy guarantees.  McSherry et al. \cite{Frank} adapted the leading approaches in the Netflix Prize competition to provide differential privacy and recommendations on movies. Machanavajjhala et al. \cite{Ash} studied recommendations based on a user's social network with differential privacy constraints. On the other side, in order to prevent a single party, e.g. the service provider, from gaining access to every user's data, cryptographic solutions are proposed in \cite{Aimeur,Canny}, however, cryptography could be computationally expensive, especially for end-users. Nandi et al. \cite{P3} proposed to preserve preference privacy from a single party by middleware, where computation and recommendation are performed locally. 

The focus of our work is to protect users from unreliable service providers, and to mitigate users' fear of potential intrusions of privacy by keeping a certain amount of anonymity. The curse of dimensionality and computational limitations of personal devices make deployment of \cite{P3} difficult. The idea of using groups as a natural protective mechanism is inspired by the French philosopher Gilbert Simondon \cite{Gilbert}.  An intriguing and interesting aspect of Simondon's theory of systems and technical objects is the idea of adopting an ``associated milieu"  into the operation of the system. This associated milieu can be natural resources.  For example, Simondon spoke of the Guimbal turbine (named after the engineer who invented it), which, to solve the problem of loss of energy and overheating, used oil to lubricate the engine and at the same time isolate it from water; it can then also integrate a river as the cooling agent of a turbine \cite{Gilbert}. The river here is the associated milieu for the technical system; it is part of the system rather than simply the environment.  Groups for us serve a similar function as an associated milieu, that contribute to the preservation of individual privacy, while still supporting the functioning of the social network. 

In this paper, we propose a framework for using groups as a natural middleware to recommend products to users. Our framework can be combined with other differentially private recommendation solutions such as \cite{Frank}. More specifically, we design a simple distributed protocol to preserve users' privacy through a peer-to-peer preference exchange process. The effective size of the anonymity set asymptotically approaches the size of the group as time approaches infinity.  After group opinion is aggregated, we construct a recommendation graph and use a random walk based method to make recommendations. The stable distribution resulting from a random walk on the graph is interpreted as a ranking of nodes for the purpose of prediction and recommendation. Personalized recommendation is only performed locally so that no private information is revealed to the service provider. We evaluate the performance of the proposed algorithm using the MovieLens  and Epinions \cite{Massa} dataset, and we compare the results with recommendation algorithms designed for individual users. 

\subsection{Contributions}
A summary of the contributions of this paper is as follows:
\begin{itemize}
\item We propose a recommendation system using groups as a natural protective mechanism for privacy preservation.  To the best of our knowledge, this is the first work to incorporate group-based social networks in recommendation systems for the purpose of protecting users' privacy. 
\item A distributed peer-to-peer preference exchange protocol is designed to guarantee anonymity. We use random walks and mixing time of Markov chains to analyze the evolution of effective size of the anonymity set with time. 
\item We suggest a novel method for intra-group preference aggregation. We propose a heuristic method based on strong connected component detection to compute Kemeny-Young ranking \cite{Kemeny}. A popularity factor is introduced to balance the quality and popularity of the ranking result.
\item We introduce a random walk based hybrid collaborative filtering graph model that incorporates group based social network information for recommendations. Experiments are designed on the MovieLens dataset to evaluate the performance of the proposed recommendation system. 
\end{itemize}

The remainder of the paper is organized as follows. We formulate the recommendation problem in Section \ref{ps}. We then introduce the group-based recommendation system in Section \ref{main}. The performance of the proposed framework is evaluated in Section \ref{experiment}, followed by conclusions in Section \ref{conclusion}. 

\begin{figure*}[!t]
\centering
\centerline{\includegraphics[width=13cm]{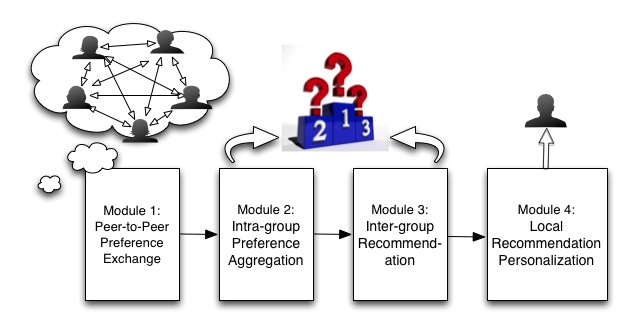}}
\caption{Modules in privacy preserving group-based recommendation system.}
\label{fig:module}
\end{figure*}

\section{Problem Statement}
\label{ps}
In a typical setting, there is a list of $m$ users $\mathcal{U} = \{u_1, u_2, ..., u_m\}$, and a list of $n$ items $\mathcal{I} = \{i_1, i_2, ..., i_n\}$.  Each user $u_j$ has a list of items $I_{u_j}$, which the user has rated or from which user preferences can be inferred.  The ratings can either be explicit, for example, on a 1-5 scale as in Netflix, or implicit such as purchases or clicks. This information is stored locally. In a group-based social network, the basic atoms are groups instead of individuals. $\mathcal{G} = \{g_1, g_2, ..., g_k\}$ is a list of $k$ groups. $\mathcal{S} = \{\mathcal{G}, \mathcal{E}_s\}$ is a group-based social network,  containing social network information, represented by an undirected or directed graph.  $\mathcal{G}$ is a set of nodes and $\mathcal{E}_s$ is a set of edges.  For all $u, v$, $(u, v) \in \mathcal{E}_s$ if $v$ is an associated group of $u$. Let $\mathcal{T} = \{t_1, t_2, ..., t_y\}$ be a set of tagging information for the items. For example, for movies, $\mathcal{T}$ can be genre, main actor, release date, etc.  $T_i \in \{0,1\}^y$ denotes the features of item $i$, where $y$ is the total number of tags. We want to make a recommendation to a group of members while no individual preference information is revealed to the central server. 

\section{Group-based Privacy Preserving Recommendation System}
\label{main}
The structure of the recommendation system is shown in Fig. \ref{fig:module}:
\begin{itemize}
\item \textbf{Module 1}: Peer-to-peer preference exchange. Users exchange preference information with other group members in a distributed manner. Only the exchanged information is then uploaded to the central node, thus the individual preferences are kept private. 
\item \textbf{Module 2}: Intra-group preference aggregation. The central server aggregates group preferences to minimize the disagreement heuristically.  The group preference will serve as an input for inter-group recommendation and prediction.
\item \textbf{Module 3}: Inter-group recommendation. A recommendation graph is constructed.  A random walk based algorithm is performed for recommendations.
\item \textbf{Module 4}: Local recommendation personalization. The top $k$ recommendations are returned to group members. Personalized recommendation are computed locally.
\end{itemize} 

In the rest of this section, we describe and analyze the system in detail.

\subsection{Peer-to-peer Preference Exchange}
\label{p2p}
Preference exchange is a process to mix individual preferences so that no full rating profile is collected by the recommendation service provider. Some of the benefits of our preference exchange scheme could be obtained by anonymous communications such as \emph{The Onion Router} \cite{tor}. Users could use persistent pseudo-identities and make anonymous ratings, either directly on the central server or let a trustful third party collect this information. However, pseudo-identities still expose users to privacy risks unless the user data is further protected \cite{Canny}. Our proposed peer-to-peer preference exchange procedure lets users exchange information within the group in a distributed manner. Only the aggregated preferences are sent to the central server. In a group based social network, such as Douban, group members are maintained by group masters,  thus we assume that users within the group are trustful and uncorrupted. Otherwise, techniques of fake accounts and malicious users detection in social networks can be used \cite{Stringhini}\cite{Yu}.  Note that the proposed P2P procedure also protects users preference information among peers, since this is beyond the scope of this work, we do not measure the privacy guarantee among users quantitatively. 

In the rest of Section \ref{p2p}, we describe our peer-to-peer preference exchange scheme in detail and analytically give the privacy guarantee towards the service provider. 

\subsubsection{Pairwise Comparison Matrix}
Before sending preference information to the server, group users exchange information with other group members distributedly. Users then upload the mixed information.  Suppose every user has a partial ranking on $\mathcal{I}$. Each user keeps an $n \times n$ pairwise comparison matrix $M$ locally. $M_{xy}^{(u)} = 1$ if user $u$ considers $x$ is better than $y$; $M_{xy}^{(u)} = 0$ if otherwise, including when no comparison is made between $x$ and $y$ or they are equally liked. When the preference information is $p$-rating records, i.e. users rate products by the scale of 1 to $p$, we can transform $p$-rating history into a partial rank. Let $r^{(u)}_x$ denote the rating of user $u$ on item $x$.  

\begin{itemize}
\item If $r^{(u)}_x > r^{(u)}_y$, $M_{xy}^{(u)} = 1$, and $M_{yx}^{(u)} = 0$.
\item If $r^{(u)}_x = r^{(u)}_y$, $M_{xy}^{(u)} = 0$, and $M_{yx}^{(u)} = 0$. 
\end{itemize}

\subsubsection{Pre-exchange Preparation}
Although our focus is to prevent the central server from collecting individual preference, the proposed P2P preference exchange scheme also protects users preference information from other group members. Before the preference exchange starts, each user $u$ randomly chooses $p$ pairwise comparison pairs $x,y$ with $M_{xy}^{(u)} = M_{yx}^{(u)} = 0$, and changes it to  $M_{xy}^{(u)} = M_{yx}^{(u)} = 1$, where
\begin{equation}
p = \frac{1}{2}\left( \frac{1}{2}n(n-1) - \sum_{i,j}\mathbf{1}_{\{M_{ij}^{(u)} + M_{ji}^{(u)} = 1\}}\right),
\end{equation}
i.e. after inserting some 1s in the pairwise comparison matrix, there are an equal number of 0s and 1s among all entries in the matrix. 

\subsubsection{Preference Exchange Rules}
\label{rules}
Although in a group-based social network, a user can belong to multiple groups, in the recommender system, each user only subscribes to one group for recommendations (If assigning users to multiple groups for recommendations, trivial changes are needed, e.g. preference aggregation on the recommendation results from multiple groups). Consider a group $g_i$ of $N$ members.  Group members form a network of $N$ nodes, labeled $1$ through $N$, which form a complete graph.  As in some distributed systems \cite{Benezit}\cite{Boyd}, each node has a clock which ticks according to a rate 1 exponential distribution.  In addition, a synchronized clock is also present at each node.  

The preference exchange phase is a process to mix individual preferences so that users do not upload anyone's full rating profile but the mixed preference of the group. The only requirement for the preference exchange is sum conservation. When a user $u$'s  local Poisson clock ticks, $u$ randomly picks another user $v$ in the same group, and randomly picks an entry in the pairwise comparison matrix $M_{xy}$ to exchange the corresponding pairwise comparison matrix entry with $v$.

This phase ends at synchronized time $t = T_{th}$. All nodes then check all pairwise comparisons: if $M_{xy} = M_{yx} = 1$, reset both entries to be 0, i.e. make $M_{xy} = M_{yx} = 0$. Then upload their current preference information to the central server. Because the information uploaded is a mixed preference, individual preference information is not provided and user privacy is protected. 

\emph{Remark}: Note that in the pre-exchange stage, changing pairwise comparison entries from 0 to 1 does not change the individual preference profile, but only to protect user's privacy from revealing to peers in the preference exchange process. 

\subsubsection{Anonymity Analysis}
\label{anonymity}
\begin{mydef}
Anonymity is the state of being not identifiable within a set of subjects, which is called the anonymity set \cite{Pfitzmann}.
\end{mydef}
One popular measurement is the notion of an \emph{anonymity set}, which was introduced for the dining cryptographers problem \cite{chaum}.  However, a rating record does not necessarily arise with equal probability from each of the group members, and so the size of the group is not necessarily a good indicator of anonymity.  Instead, we adopt an information theoretic metric for anonymity proposed in \cite{Andrei}:

\begin{mydef}
Define the effective size $\mathcal{A}$ of an anonymity probability distribution as, 
\begin{equation}
\mathcal{A} = 2^{\sum_{u \in g_i} -p_u\log_2 p_u}
\end{equation}
where $p_u$ is the probability that a rating record is from user $u$. 
\end{mydef}

In order to find the probability distribution of a certain rating record, we first analyze the random process of preference exchange. Because of the superposition property of the exponential distribution, the setup is equivalent to a single global clock with a rate $N$ exponential distribution ticking at times $\{Z_k\}_{k \ge 0}$.  The communication and exchange of preferences occurs only at $\{Z_k\}_{k \ge 0}$. 

\begin{mydef}
A random walk is a Markov process with random variables $X_1, X_2, ..., X_t, ...$ such that the next state only depends on the current state. For a random walk on a weighted graph, $X_{t+1}$ is a vertex chosen according to the following probability distribution:
 \begin{equation}
P_{ij} := P(X_{t+1} = j|X_t = i) = \frac{p_{ij}}{\sum_{j \in \mathcal{N}_i}p_{ij}},
\end{equation}
where $\mathcal{N}_i$ are the neighbors of $i$,  
$\mathcal{N}_i := \{j  | (i,j) \in \mathcal{E}\} $, and $p_{ij}$ is the weight of the edge joining node $i$ to node $j$.
\end{mydef}

Define a \emph{natural random walk $\mathcal{X}_N$} with transition matrix $P^{N} = (P_{ij})$:
\begin{itemize}
\item $P^N_{ii}=1-\frac{1}{n'N}$ for $\forall i \in \mathcal{V}$,
\item $P^N_{ij}=\frac{1}{n'N|\mathcal{N}_i|}$ for $(i, j) \in \mathcal{E}$,
\end{itemize}
where $n' $ is the number of entries exchanged in the pairwise comparison matrix, i.e., $n' = n(n-1)$, $n$ is the number of items, and $N$ is the number of members in the group.

\begin{thm}
The effective size of the anonymity set of any preference record $\mathcal{A}$ approaches the group size $N$ asymptotically with time, i.e.
\begin{equation}
\label{equ:ano2}
\lim _{t\rightarrow \infty} \mathcal{A}(t) =  N.
\end{equation} 
\end{thm}
\begin{proof}
In this random process, there are two sources stimulating the random walk from $i$ to $j$, $\forall(i,j)\in \mathcal{E}$: one is the clock of the node $i$, $P^1_{ij} = P^N_{ij}$; the other one is the clock of its neighbor $j$, $P^2_{ij} = P^N_{ji}$. Thus $P_{ij} = P^1_{ij} + P^2_{ij}$, i.e., each pairwise comparison record $\alpha$ in a node takes a \emph{biased random walk} on a complete graph, with marginal transition matrix $P = (P_{ij})$:
\begin{itemize}
\item $P_{ii}: =1-\frac{2}{N}\frac{1}{n'} $ for $\forall i \in \mathcal{V}$,
\item $P_{ij}: =\frac{1}{n'}\frac{1}{N}\frac{2}{N-1}$ for $i \ne j$,
\end{itemize}  

Hence at time $t$, the probability distribution $\mathbf{P}_t{(i)}$ of a certain  record $\alpha$ starting from node $i$ is 
\begin{equation}
\mathbf{P}_t(i) = P^t \cdot e_i,  
\end{equation}
where $e_i$ is a unit vector with value 1 on its $i$th entry, and $P$ is a symmetric stochastic matrix,
\begin{eqnarray}
\nonumber 
P = \left(\begin{array}{cccc}1-\frac{2}{N}\frac{1}{n'} & \frac{1}{n'}\frac{1}{N}\frac{2}{N-1} & \cdots & \frac{1}{n'}\frac{1}{N}\frac{2}{N-1} \\\frac{1}{n'}\frac{1}{N}\frac{2}{N-1} & 1-\frac{2}{N}\frac{1}{n'} & \cdots & \frac{1}{n'}\frac{1}{N}\frac{2}{N-1} \\\vdots & \vdots & \ddots & \vdots \\\frac{1}{n'}\frac{1}{N}\frac{2}{N-1} & \frac{1}{n'}\frac{1}{N}\frac{2}{N-1} & \cdots & 1-\frac{2}{N}\frac{1}{n'}\end{array}\right), \\
\end{eqnarray}
with eigenvalues 
\begin{equation}
\lambda_1 \ge \lambda_2 \ge \cdots \ge \lambda_N.
\end{equation}
It is a basic property of eigenvalues that the sum of all eigenvalues, including multiplicities, is equal to the trace of the matrix. It is easy to check that 
\begin{equation}
\lambda_1 = 1,
\end{equation}
\begin{equation}\label{eq:lambda}
\lambda_2  = \cdots = \lambda_N = 1 - \frac{2}{n'(N-1)}.
\end{equation}

We can express $P$ as
\begin{equation}
P = \sum_{i = 1}^N \lambda_i\mathbf{v}^T_i\mathbf{v}_i, 
\end{equation}
where the row eigenvectors $\mathbf{v}_i$ are unitary and orthogonal.
Specifically, 
\begin{equation}
\mathbf{v}_1 = (\frac{1}{\sqrt{N}}, ..., \frac{1}{\sqrt{N}}).
\end{equation}

We thus have
\begin{equation}
P^t = \sum_{i = 1}^N\lambda_i^t\mathbf{v}^T_i\mathbf{v}_i. 
\end{equation}

Notice that
\begin{equation}
\lambda_1\mathbf{v}^T_1\mathbf{v}_1 = \lambda_1^k\mathbf{v}^T_1\mathbf{v}_1 = \frac{1}{N}\mathbf{1}\mathbf{1}^T.
\end{equation}

Hence
\begin{equation}\label{eq:p}
P = \frac{1}{N}\mathbf{1}\mathbf{1}^T + \sum_{i = 2}^N \lambda_i\mathbf{v}^T_i\mathbf{v}_i. 
\end{equation}

From (\ref{eq:lambda}) to (\ref{eq:p}), we have 
\begin{eqnarray}\label{eqn:pk}
\nonumber 
P^t &=& \frac{1}{N}\mathbf{1}\mathbf{1}^T + \left(1 - \frac{2}{n'(N-1)}\right)^{t-1} \\\nonumber&\cdot&
\left(\begin{array}{cccc}1-\frac{2}{N}\frac{1}{n'} - \frac{1}{N} & \frac{1}{n'}\frac{1}{N}\frac{2}{N-1} - \frac{1}{N} & \cdots & \frac{1}{n'}\frac{1}{N}\frac{2}{N-1} - \frac{1}{N}\\\frac{1}{n'}\frac{1}{N}\frac{2}{N-1} - \frac{1}{N}& 1-\frac{2}{N}\frac{1}{n'} - \frac{1}{N} & \cdots & \frac{1}{n'}\frac{1}{N}\frac{2}{N-1} - \frac{1}{N}\\\vdots & \vdots & \ddots & \vdots \\\frac{1}{n'}\frac{1}{N}\frac{2}{N-1} - \frac{1}{N}& \frac{1}{n'}\frac{1}{N}\frac{2}{N-1} - \frac{1}{N}& \cdots & 1-\frac{2}{N}\frac{1}{n'} - \frac{1}{N}\end{array}\right).\\
\end{eqnarray}

As $t \rightarrow \infty$, each rating record $\alpha$ shows up at each node with equal probability, i.e. 
\begin{equation}
\lim _{t\rightarrow \infty} \mathbf{P}_t(i) = \frac{1}{N}\mathbf{1},
\end{equation}
for $\forall i \in \{1, 2, ..., N\}$.

Then the effective size $\mathcal{A}$ of the anonymity distribution for $\alpha$ is 
\begin{equation}
\label{equ:ano1}
\mathcal{A}(t) = 2^{- \sum _{u \in g_i} p_u(t)\log_2(p_u(t))},
\end{equation}
where $p_u(t)$ is the $u^{th}$ element in $\mathbf{P}_t(i)$. 

Moreover, we have
\begin{equation}
\label{equ:ano2}
\lim _{t\rightarrow \infty} \mathcal{A}(t) =  N.
\end{equation} 
\end{proof}

\subsection{Intra-group preference aggregation}
\label{sec:KY}
While preference aggregation has been studied extensively in the context of social choice, even the basic problem of arriving at an aggregated ranking is difficult.  One challenge is to balance the popularity (e.g., rank items according to the number of rating records) and quality (e.g., rank according to average rating). In this recommendation system, we propose to use Kemeny ranking \cite{Kemeny} as the aggregated group preference, which is a ranking that minimizes the disagreement among group members. In the rest of Section \ref{sec:KY}, we first give the definition of Kemeny top-$k$ rank, followed by a suggested heuristic method for rank aggregation.

\subsubsection{Problem Formulation}
Suppose every member has a preference profile $\pi_i$ (full ranking or partial ranking). In the recommendation system, we focus on the top-$k$ rank $\pi^k$, which is a partial rank consisting of the $k$ most popular alternatives. One way to define top-$k$ rank is that a partial rank contains $k$ items which minimizes the disagreement with all individual user's preferences, as explicitly formulated below:
\begin{equation}
\label{eq:target}
\begin{aligned}
& \underset{\pi^k}{\text{minimize}}
& &\sum_{i = 1} ^ {|g_j|} K(\pi^{k}, \pi_i) \\
\end{aligned}
\end{equation}

$ K(\pi^{k}, \pi_i)$ is the \emph{Kendall tau distance} \cite{Kendall}, defined by the number of disagreement of pairwise comparisons between two (partial) ranks. More specifically, 

\begin{equation}
K(\pi_1,\pi_2) = |\{(i,j): i < j, ( \pi_1(i) < \pi_1(j) \wedge \pi_2(i) > \pi_2(j) ) \vee  ( \pi_1(i) > \pi_1(j)  \wedge  \pi_2(i) < \pi_2(j) )\}|
\end{equation}
  
If $k$ is the size of the items, i.e. $k = n$ and $\pi^{k}$ satisfies (\ref{eq:target}), $\pi^{k}$ is called a \emph{Kemeny ranking} \cite{Kemeny}.  For example, suppose $\pi_1 = \{1,2,3\}$, $\pi_2 = \{2,1,3\}$, $\pi_3 = \{3,2,1\}$, with the pairwise comparison graph shown in Fig. \ref{fig:tri}. $K(\pi_1, \pi_2) = 1$, $K(\pi_1, \pi_3) = 2$, and the Kemeny Ranking is $\pi^3 = \{1, 2, 3\}$. Finding a Kemeny ranking is equivalent to a \emph{minimum feedback arc set} problem \cite{Karp}.

In our recommendation system, the mixed preferences are recorded in the form of pairwise comparisons. For a group $g_j$, let $\mathcal{M}^{(j)} = \sum_{i \in g_j} M^{(i)}$. We can construct a direct weighted graph $G^{(j)} = \{\mathcal{I}, E^{(j)}\}$. $(x,y) \in E^{(j)}$ if $\mathcal{M}_{xy}^{(j)} - \mathcal{M}_{yx}^{(j)}> 0$, and $w_{xy}^{(j)} =  \mathcal{M}_{xy}^{(j)} - \mathcal{M}_{yx}^{(j)}$ i.e., if more group members in $g_j$ prefer $x$ to $y$. The weight of the edge is the corresponding difference of matrix entries. In order to find the top-$k$ list $\pi ^k$ satisfying (\ref{eq:target}), we need to reverse a set of edges, the sum of which is minimal so that we can do the topological sort on the graph for the first $k$ nodes. Partial rank aggregation is known to be NP-hard \cite{Ailon}. 

\begin{figure*}[!t]
\centering
\centerline{\includegraphics[width=6cm]{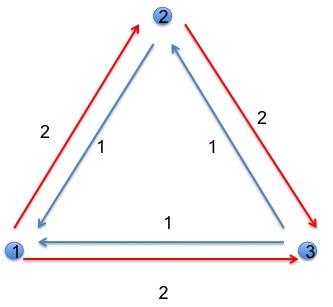}}
\caption{The pairwise comparison graph for $\pi_1 = \{1,2,3\}$, $\pi_2 = \{2,1,3\}$, $\pi_3 = \{3,2,1\}$.}
\label{fig:tri}
\end{figure*}


\subsubsection{Heuristic Rank Aggregation}
\label{sec:heuristic}

We now propose an efficient heuristic method for intra-group preference aggregation for top-$k$ items. As mentioned in the last section, if we can do topological sort in the partial rank graph for the first $k$ nodes, we then have the top-$k$ list of the group preference. We modify Tarjan's \emph{strongly connected components} (SCC) algorithm \cite{Tarjan} to find the top-$k$ list in linear time if the size of the top SCC is small compared to the size of item list $\mathcal{I}$. Since Tarjan's algorithm returns SCCs in reverse topological order, we first create the graph $G'$, the transpose graph of $G$. Let $c$ be the counter of nodes contained in the current SCC. Detection for SCCs stops when $c \ge k$. Let $\beta$ denote the maximum size of SCC popped so far. Considering the large number of items in a recommendation system, we set a threshold $\theta_{scc}$: if $\beta \ge \theta_{scc}$, a heuristic method is used to find $\pi ^k$; otherwise we compute the exact result. $k \ll \theta_{scc} \ll n$. 

In reality, the assumption that all items are equally likely to be rated may not hold. Let us define the popularity of an item $\gamma(i)$ as the percentage of users who rated item $i$. In order to balance popularity and quality, let $\theta_p$ denote the popularity threshold. An item will not be included in the top-$k$ list if $\gamma(i)  <  \theta_p$.

A summary of the algorithm is shown in Algorithm \ref{ag:all}:

\begin{algorithm}[H]
 $G' \leftarrow G^T$\;  
	\{create a graph $G'$, which is a transpose graph of $G$\}\;
$c \leftarrow 0$, $\beta \leftarrow 0$\;
\While{$c < k$}{
	TarjanSCC\;
	\{update $c$ and $\beta$ after every SCC is popped\}\;
	}
\eIf {$\beta < \theta_{scc}$}{
	topk $\leftarrow$ Kemeny\;
	}
	{
	 topk $\leftarrow$ HeuristicKemeny\;
	}
\Return topk\;
\caption{Algorithm sketch for intra-group preference aggregation.}
\label{ag:all}
\end{algorithm}

We use a modified version of TarjanSCC from \cite{Tarjan} in order to update $c$ and $\beta$. The modified SCC detection algorithm is summarized in Algorithms \ref{ag:scc1} and \ref{ag:scc2}.

\begin{algorithm}[H]
index $\leftarrow 0$\;
empty stack $S$\;
\For {$v$}{ 
	\If{$v$.index is undefined}{
		SCC($v$)\;
	}
}
\caption{SCC detection: TarjanSCC}
\label{ag:scc1}
\end{algorithm}

The function SCC recursively explores the connected nodes in the SCC, as shown in Algorithm \ref{ag:scc2}.

Much work has been done on heuristic methods for computing optimal \emph{Kendall tau distance} (Kemeny-Young method) \cite{Ailon}\cite{Lin}\cite{Claire}\cite{Saab}. In the experiments in Section \ref{experiment}, we use Borda count algorithm for HeuristicKemeny. Borda count is a 5-approximation of the Kemeny-Young method, and is often computational effective in practice \cite{Claire}. In a rating based system, the Borda count result can be calculated by adding up the rating scores of the item. However, other heuristic methods can also be integrated easily in the proposed framework. We do not discuss these methods further since it is out of the scope of this paper. 

It is easy to see that \emph{TarjanSCC} runs in linear time as a function of the number of edges and nodes because it is based on depth-first search. Borda counts runs in linear time as a function of the number of items, i.e. $O(|V|)$. We assume $k \ll \theta_{scc} \ll n$, and hence the proposed heuristic method runs in linear time in $O(|E|+|V|)$.

\begin{algorithm}[H]
$v$.index $\leftarrow$ index\;
$v$.root $\leftarrow$ index\;
index $\leftarrow$ index $+ 1$\;
$S$.push($v$) \;
 \For {$(v, w) \in $ edges of $G'$}{
 	\If{w.index is undefined}{
		 SCC$(w)$\;
		$v$.root $\leftarrow \min(\rm{v.root, w.root})$\;
	}
	\If{ $w \in $current\_$s$}{
		$v$.root $\leftarrow \min(\rm{v.root, w.index}))$\;
	}
}
 \If{$v$.root = $v$.index}{
	 empty stack current\_$s$\;
	\Repeat	{$u = v$}	{
	 $u \leftarrow S.\rm{pop}()$\;
	\If{popularity(u) > $\theta_p$}{
		current\_$s$.push$(u)$\;
	}
	}	

	 output current\_$s$\;
	 $c \leftarrow c + $current\_$s\rm{.size}()$\;
	\If{current\_$s.size() > \beta$}{
		$\beta \leftarrow $ current\_$s\rm{.size}()$\;
	}
	\If{$c > k$}{
		exit\;
	}
	}
 \caption{Function SCC}
 \label{ag:scc2}
 \end{algorithm}

\subsection{Inter-group Recommendation}

Intra-group preference aggregation described above gathers existing preference information from group members. However, it is desirable to recommend new items that have similar features but that have not yet been rated by group members. Studies show that two individuals connected via a social relationship tend to have similar tastes, which is known as the ``homophily principle" \cite{He}. With the absence of individual preference records, a group preference can serve as a natural middleware to help make recommendation decisions while protecting the privacy of users.

An intuitive approach is collaborative filtering (CF) \cite{Bell}\cite{Su}\cite{Vucetic}. Collaborative filtering is one of the most successful approaches to building a recommendation system. It uses the known preferences of users to make recommendations or predictions to a target user \cite{Su}. Weighted sum is typically used to make predictions.

In CF, a generally adopted similarity measure is called \emph{Pearson Correlation} which measures the extent to which two variables linearly relate with each other \cite{Resnick}. For user-based algorithms, the Pearson Correlation between user $u$ and $v$ is
\begin{equation}
\label{}
\displaystyle
w_{u,v} = \frac{\sum_{i \in I}(r_{u,i} -\bar{r}_u)(r_{v,i} -\bar{r}_v)}{\sqrt{\sum_{i \in I}(r_{u,i}-\bar{r}_u)^2}\sqrt{\sum_{i \in I}(r_{v,i}-\bar{r}_v)^2}},
\end{equation}
where $i \in I$ is an item rated by both users $u$ and $v$, $r_{u,i}$ is the rating of user $u$ on item $i$, and $\bar{r}_u$ is the average rating of user $u$ in the co-rating set $I$.  A weighted sum is then taken to predict the rating for target user $u$ on a certain item $i$ \cite{Resnick}
\begin{equation}
\label{}
\displaystyle
R_{u,i}= \bar{r}_u + \frac{\sum_{v \in \mathcal{U}}(r_{v,i} - \bar{r}_v) \cdot w_{u,v}}{\sum_{v \in \mathcal{U}}|w_{u,v}|}. 
\end{equation}

Recommenders based on collaborative filtering then refer to this prediction to provide the top-$k$ recommendations to the user. For our group-based recommendation, we can treat the groups as users in the equations above, and use the aggregated group preference as the rating history.  In this way, a group recommendation could be made.

However, traditional collaborative filtering methods are challenged by problems such as \emph{cold start} and \emph{data sparsity}. In the case of a group based recommendation system, these problems are inevitable, especially since groups in a social network already form natural clusters.  Hence, there may not be many co-rated items between different groups for the Pearson Correlation computation.

In order to overcome the disadvantages of collaborative filtering, we propose a random walk based inter-group recommendation system, which is an extension of our previous work in \cite{Shang}. Our model incorporates content information of items and social information of groups together as group preference information. It is shown in \cite{Kleinberg} that a random walk approach is very effective in link prediction on social networks. Inspired by \cite{Page} and \cite{Kleinberg}, we create a recommendation graph, as shown in Fig. \ref{fig:model}, consisting of items, groups, and item genres as nodes. Similar to PageRank, the stable distribution resulting from a random walk on the recommendation graph is interpreted as a ranking of the nodes for the purpose of recommendation and prediction. We describe how to construct this recommendation graph and represent the flow on the graph in the rest of this section.

\begin{figure}[!t]
\centering
\centerline{\includegraphics[width=7cm]{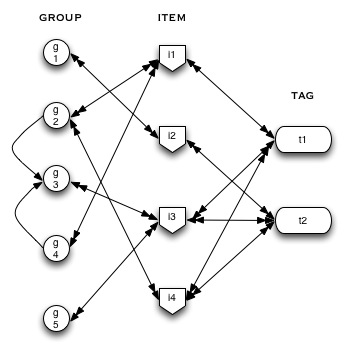}}
\caption{Example of a recommendation graph for inter-group recommendations.}
\label{fig:model}
\end{figure}

\subsubsection{Graph settings}
Let $G = \{\mathcal{V}, \mathcal{E}\}$ be a graph model for a recommendation system, where $\mathcal{V}:= \mathcal{G}\cup \mathcal{I} \cup \mathcal{T}$. The nodes of the graph consist of groups, items and item information. For $v_i, v_j \in \mathcal{V}$, $(v_i, v_j)\in \mathcal{E}$ if and only if there is an edge from $v_i$ to $v_j$, which is determined as given below. The weights are specified in the next subsection.
\begin{itemize}
\item For $g\in \mathcal{G}, i \in {\mathcal{I}}$, $(g, i) \in \mathcal{E}$ and $(i, g) \in \mathcal{E}$ if and only if $i \in \pi^{k}(g)$.  i.e., an item $i$ and a group $u$ are connected with weights $w_{gi}$ and $w_{ig}$ if $i$ is in $g$'s top-$k$ list. 
\item For $i \in \mathcal{I}, t\in \mathcal{T}$, $(i, t)\in \mathcal{E}$ and $(t, i)\in \mathcal{E}$ if and only if $T_{i}^{(t)} \ne 0$.  i.e., an item $i$ and tag $t$ are connected with weights $w_{it}$ and $w_{ti}$ if $i$ is tagged by $t$.
\item For $g_1, g_2 \in \mathcal{G}$, $(g_1, g_2)\in \mathcal{E}$ with weight $w_{g_1g_2}$ if and only if $g_1, g_2$ are associated groups, i.e. $(g_1, g_2)\in \mathcal{E}_s$, as mentioned in Section \ref{ps}. 
\end{itemize}

\subsubsection{Edge weight assignment}
\label{sec:weight}
The main part of our rank graph is the collaborative filtering graph, which includes the group nodes, item nodes, and the edges between them. One way to assign weights on the collaborative filtering graph is by setting
\begin{equation}
w_{gi} = w_{ig} = \frac{k + 1 - \pi_g^k(i)}{k}w_{\max},
\label{eq:weight}
\end{equation}
where $\pi_g^k(i)$ is the rank of item $i$ in the top-$k$ list of group $g$, and $w_{\max}$ is the max weight assigned on the graph. Let $\pi_g^k(i) = k+1$ if $i \notin {\pi_g^k}$. 

Note that a larger edge weight indicates greater chance that the random walk passes through that edge. An item $i$ with better rank in $\pi_g^k(i)$ results in larger weights on edges involving $i$. 

For the extended graph, i.e. nodes and edges containing item content, group social network information, etc., we simply assign an edge weight of 1 if an edge is present.

 \subsubsection{Rank Score Computation}
 \label{sec: RSC}
For the recommendation graph $G = \{\mathcal{V}, \mathcal{E}\}$. Let $v = |\mathcal{V}|$ denote the number of nodes on the graph.  $m$ is a $v\times 1$ \emph{customized probability vector}. 
\begin{equation}
\theta = e_g,
\end{equation}
where $e_1, e_2, ..., e_v$ are the standard basis of column vectors. $\beta$ is a \emph{damping factor}. With probability $1-\beta$, the random walk is teleported back to node $g$. The rank score $s$ satisfies the following equation:
\begin{equation}
 s = \beta Ws + (1- \beta)\theta,
 \label{pagerank}
 \end{equation}
  where $W$ is the weighted transition matrix with $W_{ij} = P_{ji}$. 
 
 So we have,
 \begin{equation}
s = \big(\beta W + (1 - \beta)\theta\textbf{1}^T\big)s := Ms
 \end{equation}
Hence the rank score is the \emph{principal eigenvector} of $M$, which can be computed by iterations fast and easily via Algorithm \ref{ag:pr}.

\begin{algorithm}
$s^{(0)}_i  \leftarrow \frac{1}{v}$ for all $i$\; 
$t = 1$\; 
\While{ $|s^{(t)} - s^{(t-1)}| < \epsilon$} {
	\For {$i = 1$ to $v$}{
	$s^{(t)}_i = \sum_{j = 1}^v \beta W_{ij}s^{(t-1)}_i + (1 - \beta)\theta_i$\;
	}
	$t \leftarrow t + 1$\;
}
\caption{Iterative computation of rank score}
\label{ag:pr}
\end{algorithm}

The rank score $s$ can be interpreted as the importance of other nodes to the target group $g$. It is easy to see that we can increase the rank score by shortening the distance, adding more paths, or increasing the weight on the path to $g$. These are desirable properties in a recommendation system. For example, even if item $i$ is not directly connected with $g$, but it is in a category to which many of $g$'s top-$k$ items belong, then $i$ is very likely to have a high rank score. Or if group $g$ and $g'$ have many overlapping top-$k$ items, $g'$ will have high rank, so we can use $g'$'s top-$k$ list to make recommendations and predictions for $g$.


\subsubsection{Recommendations}

\textbf{Direct Method:}
Solving Equation (\ref{pagerank}) iteratively, we obtain a rank score for all nodes of the recommendation graph $G$. Since the rank score represents the importance to the target group, we can then separate and sort them according to the categories, i.e. groups $\mathcal{G}$, items $\mathcal{I}$, tags $\mathcal{T}$, etc.  Sorted items form a recommendation list to the target group $g$, and we can compute the recommendation for every group.

\textbf{User-based Prediction}
For items above the group popularity threshold, we simply take the average rating of group members as the rating prediction. For other items, we can use rank score as an influence measure to make predictions, which is similar to memory-based collaborative filtering, using \emph{Pearson Correlation} \cite{Resnick} as a similarity measure between users and items. Given the rank score of the group set $\mathcal{G}$, we take the weighted sum of the groups' ratings on item $i$ as a prediction for the target group $g$,  as shown below:
\begin{equation}
\hat{r}_{gi}^{user} = \frac{\sum_{x\in G_i}s_x(\bar{r}_{xi} - \bar{r}_x)}{\sum_{x\in G_i}s_x} + \bar{r}_x.
\label{user}
\end{equation} 
$G_i$ is the set of groups for which item $i$ is above the popularity threshold. $s_x$ is the target group's personalized rank score of group $x$.  

\textbf{Item-based Prediction}
As above, in order to perform an item-based recommendation, we can use the rank score of item set $\mathcal{I}$ as weight to predict the rating of the item $i$ for the target group $g$, if the popularity of the item is below the threshold.  Specifically,
\begin{equation}
\hat{r}_{gi}^{item} = \frac{\sum_{j\in I_g}s_jr_{gj}}{\sum_{j\in I_g}s_j}.
\label{item}
\end{equation} 
In Equation (\ref{item}), we use $u$'s rating on similar items to predict the rating on $i$. $s_j$ is the target group's personalized rank score of item $j$. 

After a recommendation is made, results are returned to individual users. Items that have been rated by the user, which are stored locally, are then removed from the recommendation list.

\section{Experiments and Evaluation}
\label{experiment}
\subsection{Dataset}
In order to evaluate the performance of the proposed algorithm, we run experiments on the MovieLens and Epinions dataset, both of which are widely used benchmarks for recommendation systems. The MovieLens dataset consists of 1,682 movies and 943 users. Movies are labeled by 19 genres. User profile information such as age, gender, and occupation is also available. In order to evaluate the group-based recommendation system, we take user profile categories provided in the dataset as groups. In the experiments, we group users in three different ways, namely, gender, age, and occupation. Detailed group category distribution is as follows:
\begin{itemize}
\item \textbf{Gender}: male (71.16\%) and female (28.84\%).
\item \textbf{Age}: below 21, 21 to 30, 31 to 40, 41 to 50, above 50, indexed from 1 to 5, respectively, as shown in Fig.\ref{fig: ga}.
\item \textbf{Occupation}: administrator, artist, doctor, educator, engineer, entertainment, executive, healthcare, homemaker, lawyer, librarian, marketing, none, other, programmer, retired, salesman, scientist, student, technician and writer, indexed from 1 to 21, respectively, as shown in Fig.\ref{fig:oc}.
\end{itemize}

Epinions is a website where users can post their reviews and ratings (1-5) on a variety of items (songs, softwares, TVs, etc.), as long as user's  \emph{web of trust}, i.e. ``reviewers whose reviews and ratings they have consistently found to be valuable" \cite{Massa}. We randomly select 946 items, 304 users and their trust network from Epinion dataset to perform the experiments. Using the community detection techniques in \cite{blondel2008fast}, we detected 18 groups based on the trust network.


\begin{figure*}[!t] 
\centering
\subfloat[Percentage of the population of 5 different age categories.  ]{\label{fig: ga}\includegraphics[width=0.48\textwidth]{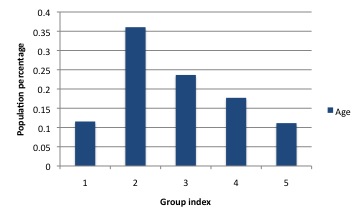}}             
\subfloat[Percentage of the population of 21 different occupation categories.]{\label{fig:oc}\includegraphics[width=0.48\textwidth]{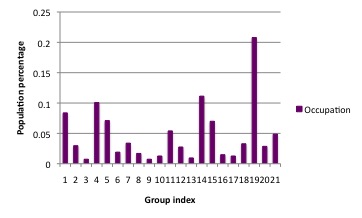}}
 \caption{The group distribution of MovieLens datasets.}
 \label{fig:group}
\end{figure*}

\subsection{Experimental Methodology and Results}
\label{sec:method}
We evaluate our results with two popular evaluation metrics for top-$k$ recommendations: percentile and recall. 

\emph{Percentile}: The individual percentile score is simply the average position (in percentage) that an item in the test set occupies in the recommendation list. For example, if four items are ranked 1st, 9th, 10th and 20th in a recommendation list consisting of 100 items, with individual percentile scores of 0.01, 0.09, 0.10 and 0.20. The average percentile of the system is 0.1. A lower percentile indicates a better prediction.

\emph{Recall}: Given a recommendation test, we consider any item in the top-$k$ recommendations that matches any item in the test set as a ``hit", as in \cite{Tso}. 
\begin{equation}
recall(k) = \frac{\#\textrm{hits of top-k}}{T},
\end{equation} 
where $T$ is the size of test set. A higher recall value indicates a better prediction.

In this experiment,  all items in the test set $T$ are rated 5 (highest rating) by users, thus we can consider them as relevant items for recommendation.  The recommendation list has a length of 900 items for MovieLens dataset and 500 for Epinions dataset. The top-500 movies in the aggregated group preference list are used to construct the recommendation graph for MovieLens and top-300 for Epinions. Note that the popularity threshold of the recommendation system can be decided by users, since different groups may have a different requirement for popularity. In our experiment, we set the popularity threshold at 0.01.

We compare the proposed method with two state-of-art personalized recommendation systems:  L+ \cite{Fouss} and ItemRank \cite{Gori}. L+ suggested a dissimilarity measure between nodes of a graph, the expected commute time between two nodes, which the authors applied to recommendation \cite{Fouss}.  Specifically, they constructed a non-directed bipartite graph where users and movies form the nodes. A link is placed between a user and movie if the user watched that movie. Movies are then ranked in ascending order according to the average commute time to the target node. ItemRank built the recommendation graph by only using movies as nodes. In \cite{Gori}, two nodes are connected if at least one user rated both nodes. The weight of the edge is set as the number of users who rated both of the nodes. A random-walk based algorithm is then used to rank items according to the target user's preference record. In order to see how much information is lost by grouping users, we also compare the proposed privacy-preserving recommendation algorithm with a recommendation graph of similar structure, but with all the individual rating information, where nodes of the recommendation graph are formed by users, items, user social profile information (gender, age and occupation). The weight of an edge between users and items is given by
\begin{equation}
w_{ui} = w_{iu} = \exp\left({\frac{r_{ui} - \bar{r}_u}{\sqrt{\sum_{i\in I_{u}}(r_{ui} - \bar{r}_u)^2}}}\right),
\end{equation}
\begin{equation}
\bar{r}_u := \frac{\sum_{i\in I_u}r_{ui}}{|I_u|}.
\end{equation}
where $I_{u}$ denotes the set of items which user $u$ has rated. Note that a larger edge weight indicates more chance that the random walk passes through that edge. If user $u$'s rating on item $i$ $r_{ui}$ is lower than the average rating $\bar{r}_u$, $w_{ui}$ and $w_{iu}$ are less than 1; otherwise are greater than 1. The assignment of weights do not depend on the variance of the user's ratings.

Experimental results of cross-validation on percentile scores of the MovieLens dataset are shown in Table 1. We create five training/testing splits. Although it does not utilize knowledge of individual's preference information, the proposed group-based privacy preserving recommendation algorithm still has a better performance than L+ and ItemRank in both datasets, which are two state-of-art personalized recommendation methods. And as expected, due to the absence of personal rating information, the performance of the proposed group method is inferior to \emph{personal recommendation}, i.e., recommendations with individual rating information. It is worth noting that in MovieLens dataset, among all three different ways of grouping users, grouping by occupation outperforms the other two grouping methods, which shows the promise of group-based recommendation system with finer groups. Moreover, in order to evaluate the effectiveness of groups in MovieLens dataset, we did contrast experiments on random groups, which are users divided randomly into 2, 5, 21 groups to compare with gender, age and occupation groups. Experimental results show that the natural groups outperform the random groups, as shown in Table 1. We also randomly assigned users of Epinions dataset to 18 groups. Surprisingly, the random groups perform slightly better than clusters from community detection. However, this result agrees with the experiments on personal recommendation in Table 1, the percentile of the personal recommendation without social network information is better than that with social network information.
 
  \begin{table}[!t]
\begin{center}
\caption{Average percentile results obtained by 5-fold cross-validation for recommendation.} 
\begin{tabular}{llll}
\hline
\multicolumn{2}{ c }{MovieLens}&\multicolumn{2}{ c }{Epinions} \\
\hline
Methods & Percentile & Methods & Percentile\\
\hline
L+ &  0.1157 & L+ & 0.4023\\
ItemRank &   0.1150 &ItemRank  & 0.4156\\
Personal Recom. w/ SN info  & \textbf{0.0790}  &Personal Recom w/ SN & 0.2444\\
Personal Recom. w/o SN info  & 0.0813  &Personal Recom w/o SN & \textbf{0.2311}\\
Group by Gender &  0.1110 & Group by comm. detection & 0.3752\\
Group by Age & 0.1066 & Random 18 groups & \textbf{0.3689}\\  
Group by Occupation  & \textbf{0.1060}& &\\
Random 2 Groups & 0.1172 & & \\
Random 5 Groups & 0.1149 & & \\
Random 21 Groups & 0.1104& & \\
\hline
\end{tabular}
\end{center}
\label{tb:percentile}
\end{table}

We also perform 5-fold cross-validation experiments for recall values, as shown in Table 2 and Table 3. In real settings, a user is unlikely to browse a very long recommendation list.  Thus, we only test the top-5 to top-50 recall values. As introduced in Section \ref{sec:method}, a recall value of $k$ is the probability that an item in the test set hits the top-$k$ items recommended by the system. A higher recall value means a higher chance that items in the test set appear in the top-$k$ list. Since these items all have the highest ratings, a higher recall value indicates better performance of the recommendation algorithm. Table 2 shows the results from MovieLens dataset. \emph{Personal recommendation},  our proposed algorithm with individual preference information, trading privacy for quality, has the best performance. Otherwise, L+ has better performance on top-5 recall, and the recommendation system based on occupation groups outperforms gender and age groups, and also has a higher recall value than L+ and ItemRank for top-10 to top-50 recommendations. Similar results from Epinions dataset are shown in Table 3.

 \begin{table*}[!t]
\begin{center}
\caption{Average recall results obtained by 5-fold cross-validation for recommendation on MovieLens dataset.} 
\begin{tabular}{lcccccccccc}
\hline
Methods & Top-5 & Top-10 & Top-15 &Top-20 & Top-25 & Top-30 & Top-35 & Top-40 & Top-45 & Top-50\\
\hline
L+ &  \textbf{0.157} & 0.234 & 0.278 & 0.317 & 0.352 & 0.377 & 0.412 & 0.435 & 0.4601 & 0.481\\
ItemRank &   0.169 & 0.233 & 0.285 & 0.335 & 0.379 & 0.408 & 0.436 & 0.458 & 0.484 & 0.504\\
Personal Recommendation   &  \textbf{0.219} & \textbf{0.303} & \textbf{0.348} & \textbf{0.416} & \textbf{0.460} & \textbf{0.491} & \textbf{0.514} & \textbf{0.546} & \textbf{0.571} & \textbf{0.591}\\
Group by Gender &  0.104 & 0.166 & 0.244 & 0.313 & 0.366 & 0.408 & 0.442 & 0.470 & 0.489 & 0.510\\
Group by Age & 0.126 & 0.228 & 0.286 & 0.333 & 0.377 & 0.411 & 0.442 & 0.469 & 0.492 & 0.514\\
Group by Occupation  &0.149 & \textbf{0.240} & \textbf{0.305} & \textbf{0.348} & \textbf{0.386} & \textbf{0.421} & \textbf{0.449} & \textbf{0.473} & \textbf{0.496} & \textbf{0.518}\\
\hline
\end{tabular}
\end{center}
\label{tb:recall}
\end{table*}

\begin{table*}[!t]
\begin{center}
\caption{Average recall results obtained by 5-fold cross-validation for recommendation on Epinions datasets.} 
\begin{tabular}{lcccccccccc}
\hline
Methods & Top-5 & Top-10 & Top-15 &Top-20 & Top-25 & Top-30 & Top-35 & Top-40 & Top-45 & Top-50\\
\hline
L+ 		&  0.041 &0.074 & 0.098 & 0.125 & \textbf{0.159} & 0.180 & 0.202 & 0.219 & 0.237 & 0.248\\
ItemRank &   0.042 & \textbf{0.080} & 0.105 & 0.127 & 0.154 & 0.179 & 0.196 & 0.219 & 0.232 & 0.245\\
Personal Recommendation   &  \textbf{0.093} & \textbf{0.158} & \textbf{0.181} & \textbf{0.219} & \textbf{0.246} & \textbf{0.270} & \textbf{0.297} & \textbf{0.318} & \textbf{0.345} & \textbf{0.361}\\
Group by Community Detection  &\textbf{0.045} & 0.075 & \textbf{0.108} & \textbf{0.129} & 0.158 & \textbf{0.182}& \textbf{0.205} & \textbf{0.220} & \textbf{0.239} & \textbf{0.250}\\
\hline
\end{tabular}
\end{center}
\label{tb:recall}
\end{table*}

\section{Conclusions}
\label{conclusion}
In this paper, we present a framework for group-based privacy preserving recommendation systems. We introduce the novel idea of using groups as a natural protective mechanism to preserve individual users' private preference data from the central service provider. A distributed peer-to-peer preference exchange process is designed to provide anonymity of group members. We also introduce a hybrid recommendation model based on random walks. It incorporates item content and group social information to make recommendations for groups. Personalized recommendations are made locally to group members, so that no user preference profile is leaked to the service provider. Experimental results on MovieLens and Epinions datasets show that the proposed algorithm outperforms the baseline algorithms L+ and ItemRank, despite the absence of personal preference information. Thus, using our group-based method, we can obtain excellent recommendation performance while simultaneously preserving privacy.

\bibliographystyle{abbrv}
\bibliography{myrefs}

\end{document}